\documentclass[11pt]{article}

\usepackage[a4paper,margin=1in]{geometry}
\usepackage{amsmath,amssymb,amsthm,mathtools}
\usepackage{enumitem}
\usepackage{xcolor}
\usepackage{microtype}
\usepackage{hyperref}

\hypersetup{
  colorlinks=true,
  linkcolor=blue!50!black,
  citecolor=blue!50!black,
  urlcolor=blue!50!black,
  pdftitle={A Deterministic (2+epsilon)-Approximation for Weighted Feedback Vertex Set in Tournaments},
  pdfauthor={Hanqing Li, Zihan Wu}
}

\newtheorem{theorem}{Theorem}[section]
\newtheorem{lemma}[theorem]{Lemma}

\newtheorem{corollary}[theorem]{Corollary}
\theoremstyle{definition}
\newtheorem{definition}[theorem]{Definition}

\newcommand{\TFVS}{\mathsf{TFVS}}
\newcommand{\OPT}{\operatorname{OPT}}
\newcommand{\tauv}{\tau}
\newcommand{\mcH}{H}
\newcommand{\Rk}{R_k}
\newcommand{\Pk}{P_k}
\newcommand{\Lk}{L_k}

\title{A Deterministic $(2+\varepsilon)$-Approximation for Weighted Feedback Vertex Set in Tournaments}
\author{Hanqing Li\quad Zihan Wu\\[2pt]\small Peking University}
\date{}

\begin{document}
\maketitle

\begin{abstract}
We study the weighted feedback vertex set problem in tournaments.  For every
fixed integer $k\geq 2$, we give a deterministic $(2+1/k)$-approximation
algorithm with running time $n^{2^{O(k)}}$, apart from polynomial dependence
on the encoding length of the weights.  Consequently, for every fixed
$\varepsilon>0$, weighted feedback vertex set in tournaments has a
deterministic $(2+\varepsilon)$-approximation running in time
$n^{2^{O(1/\varepsilon)}}$.  The algorithm combines two ingredients.  When the
triangle graph of the tournament has bounded clique number, a chain
decomposition of its transitive complement yields an exact dynamic program for
a maximum-weight transitive subtournament.  When the clique number is large, a
structural theorem for triangle graphs supplies a constant-size strongly good
cost vector.  A local-ratio reduction with this cost vector gives the claimed
guarantee.  As a by-product, the dynamic program solves weighted feedback
vertex set exactly in $\mathcal B_7$-free tournaments in time $O(n^7)$, where
$\mathcal B_7$ is the family of seven-vertex tournaments with feedback vertex
set number at least three.
\end{abstract}

\section{Introduction}

A tournament is an orientation of the complete graph.  The feedback vertex
set problem asks for a minimum-weight set of vertices whose deletion makes the
digraph acyclic.  We call the resulting problem weighted feedback vertex set in
tournaments, or weighted $\TFVS$.

The special structure of tournaments makes $\TFVS$ a clean structured
hitting-set problem: a tournament is acyclic if and only if it is transitive,
and a tournament is cyclic if and only if it contains a directed triangle.
The deterministic approximation ratio was successively improved from $5/2$
by Cai, Deng, and Zang~\cite{cai2001} to $7/3$ by Mnich, Vassilevska Williams,
and V\'{e}gh~\cite{mnich2016}; see also the simpler $7/3$ algorithm of Aprile
et al.~\cite{aprile2023}.  Ghorbani and Mnich recently obtained a deterministic
$9/4$-approximation, even for the larger class of quasi-transitive
digraphs~\cite{ghorbani2026}.  Lokshtanov et al.~\cite{lokshtanov2021} gave a
factor-two approximation for weighted $\TFVS$ that runs in randomized
polynomial time; enumerating its pivot choices gives a deterministic
quasi-polynomial-time algorithm.

This paper gives, to the best of our knowledge, the first deterministic
polynomial-time approximation guarantee that can be made arbitrarily close to
two.  For $k=4$ our ratio matches $9/4$, and for every $k\geq 5$ it improves
that deterministic polynomial-time bound.  The price is a large dependence on
the accuracy parameter: the exponent is double exponential in
$1/\varepsilon$.  Thus the contribution is structural and complexity
theoretic rather than practical.

Our starting point is the triangle graph.  Its vertices are the tournament
vertices, and two vertices are adjacent when their tournament arc belongs to a
directed triangle.  Ghorbani and Mnich~\cite{ghorbani2026} prove that the
complement of the triangle graph, oriented by the tournament, is transitive.
This gives a poset and permits a chain decomposition.  They also prove that a
sufficiently large transitive tournament inside the triangle graph forces a
constant-size tournament with a large feedback vertex set.

Our main technical observation is that a chain decomposition of the complement
does more than certify bounded width: it leads to an exact dynamic program.
The state records the last selected vertex on each chain.  Relative to any
fixed vertex, the vertices of a complement chain that point to it form a
prefix.  Consequently, two partial solutions with the same last vertices have
exactly the same possible extensions.

The remainder of the paper is organized as follows.
Section~\ref{sec:prelim} introduces the terminology and the external structural
theorem.  Section~\ref{sec:dp} proves the exact bounded-width dynamic program.
Section~\ref{sec:obstruction} constructs the constant-size strongly good cost
vector.  Section~\ref{sec:algorithm} gives the local-ratio algorithm and proves
its approximation ratio.  Section~\ref{sec:complexity} analyzes the running
time and derives the $(2+\varepsilon)$ statement.

\section{Preliminaries}\label{sec:prelim}

\subsection{Tournaments and feedback vertex sets}

\begin{definition}[Tournament]
A tournament $T$ is a digraph such that for every two distinct vertices $u,v$,
exactly one of $u\to v$ and $v\to u$ is present.  For $X\subseteq V(T)$,
$T[X]$ denotes the induced subtournament.
\end{definition}

\begin{definition}[Feedback vertex set]
A set $S\subseteq V(T)$ is a feedback vertex set (FVS) if $T-S$ is acyclic.
Throughout, an input weight function is
$w\colon V(T)\to\mathbb{Q}_{\geq 0}$, encoded in binary.  We write
$w(S)=\sum_{v\in S}w(v)$ and
\[
  \OPT(T,w)=\min\{w(S):S\text{ is an FVS of }T\}.
\]
For an unweighted tournament $U$, let $\tauv(U)$ denote the minimum cardinality
of an FVS of $U$.
\end{definition}

\begin{definition}[Transitive tournament]
A tournament is transitive if its vertices admit an ordering
$v_1,\ldots,v_m$ such that $v_i\to v_j$ whenever $i<j$.
\end{definition}

Every acyclic tournament has a unique topological ordering and is transitive.
Conversely, a transitive tournament is acyclic.  Also, every directed cycle in
a tournament contains a directed triangle.  Indeed, take a shortest directed
cycle.  If it has length at least four, the tournament arc between its first
and third vertices yields either a shorter directed cycle or a directed
triangle.  Hence an FVS in a tournament is exactly a vertex set hitting every
directed triangle.

For $k\geq2$, let $\mathcal B_{2k+1}$ be the family of tournaments $U$ on
$2k+1$ vertices such that $\tauv(U)\geq k$.  We say that a tournament is
$\mathcal B_{2k+1}$-free if it has no induced subtournament in this family.
This is the family denoted $\mathcal T_{2k+1}$ in~\cite{ghorbani2026}; we use
$\mathcal B$ to distinguish a family of obstructions from individual named
tournaments.

\subsection{The triangle graph}

\begin{definition}[Triangle graph]
The triangle graph $\mcH(T)$ is the undirected graph with vertex set $V(T)$ in
which $uv$ is an edge if and only if the tournament arc between $u$ and $v$
belongs to some directed triangle of $T$.  When needed, every edge of
$\mcH(T)$ is oriented according to the corresponding arc of $T$.
\end{definition}

Let $\overline{\mcH}(T)$ denote the directed complement: it contains exactly
those tournament arcs that belong to no directed triangle, with their original
orientation.  We use the following two results as a black box.  They are
Lemma~6 and Theorem~10, respectively, of Ghorbani and
Mnich~\cite{ghorbani2026}, specialized to tournaments.

\begin{theorem}[Ghorbani--Mnich structural theorem]\label{thm:blackbox}
For every tournament $T$:
\begin{enumerate}[label={\rm(\roman*)}]
  \item $\overline{\mcH}(T)$ is a transitive digraph.  In particular, its arcs
  define a strict partial order, and the underlying graph of $\mcH(T)$ is
  perfect.
  \item For every integer $k\geq 2$, if the oriented graph $\mcH(T)$ contains
  a transitive tournament on $2^{k-1}+1$ vertices, then $T$ contains a member
  of $\mathcal B_{2k+1}$.
\end{enumerate}
\end{theorem}

For a fixed $k$, define
\[
  \Rk=\frac{2k+1}{k},\qquad
  \Pk=2^{k-1},\qquad
  \Lk=\left\lceil\frac{(2k+1)\Pk}{k+1}\right\rceil.
\]

\subsection{Local ratio and strongly good costs}

We use the following elementary form of local ratio, originating in the method
of Bar-Yehuda and Even~\cite{baryehuda1981}.  Suppose
$w=w'+\Delta c$, where $w'$ and $c$ are nonnegative cost vectors and
$\Delta\geq0$.  If a feasible output $S$ satisfies
\[
  w'(S)\leq \rho\OPT(T,w')
  \quad\text{and}\quad
  c(S)\leq \rho\OPT(T,c),
\]
then $w(S)\leq\rho\OPT(T,w)$.  Indeed, for an optimum solution $O$ under $w$,
\[
 \OPT(T,w')+\Delta\OPT(T,c)
 \leq w'(O)+\Delta c(O)=w(O).
\]

\begin{definition}[Strongly good cost]
A nonnegative cost vector $c$ is strongly $\rho$-good if every feasible FVS
$S$ satisfies
\[
  c(S)\leq \rho\OPT(T,c).
\]
\end{definition}

Strong goodness is stronger than the usual local-ratio condition on minimal
solutions.  It is convenient here because the recursive algorithm may return
any feasible FVS.

\section{Exact solution at bounded triangle width}\label{sec:dp}

We first prove the exact subroutine used when $\mcH(T)$ has small clique
number.

\begin{lemma}\label{lem:dp}
If $\omega(\mcH(T))\leq b$, then a minimum-weight FVS of $T$ can be computed in
time
\[
  O\bigl(bn(n+1)^b\bigr).
\]
In particular, for fixed $b$ the running time is $O_b(n^{b+1})$.
\end{lemma}

\begin{proof}
By Theorem~\ref{thm:blackbox}(i), the oriented complement
$\overline{\mcH}(T)$ is a poset.  An antichain in this poset is exactly a
clique in $\mcH(T)$, so its width is at most $b$.  Dilworth's
theorem~\cite{dilworth1950} gives a partition into $r\leq b$ chains.  Write
\[
 C_i=(v_{i,1},v_{i,2},\ldots,v_{i,n_i}),
 \qquad v_{i,p}\to v_{i,q}\text{ for }p<q.
\]

Every arc inside a chain belongs to $\overline{\mcH}(T)$ and hence to no
directed triangle of $T$.  Fix a vertex $u$ outside $C_i$ and indices $p<q$.
It is impossible that $u\to v_{i,p}$ and $v_{i,q}\to u$, since then
\[
  u\to v_{i,p}\to v_{i,q}\to u
\]
would be a directed triangle.  Thus the vertices of $C_i$ that point to $u$
form a prefix.  In particular, if $v_{i,q}\to u$, then
$v_{i,p}\to u$ for every $p<q$.

We compute a maximum-weight transitive subtournament.  A state is a vector
\[
  a=(a_1,\ldots,a_r),\qquad 0\leq a_i\leq n_i,
\]
where $a_i$ is the index of the last selected vertex of chain $C_i$;
$a_i=0$ means that no vertex of $C_i$ has been selected.  Let $F(a)$ be the
maximum weight of a transitive sequence realizing state $a$, with value
$-\infty$ when no such sequence exists.  Initialize
\[
  F(0,\ldots,0)=0.
\]

For every reachable state $a$ and every $t>a_i$, we may append $v_{i,t}$ if
\begin{equation}\label{eq:transition-test}
  v_{j,a_j}\to v_{i,t}
  \qquad\text{for every }j\text{ with }a_j>0.
\end{equation}
If $a^{(i,t)}$ is obtained from $a$ by replacing $a_i$ by $t$, the update is
\begin{equation}\label{eq:dp-update}
  F\bigl(a^{(i,t)}\bigr)
  \gets
  \max\left\{F\bigl(a^{(i,t)}\bigr),F(a)+w(v_{i,t})\right\}.
\end{equation}
The state graph is acyclic because every transition strictly increases
$\sum_i a_i$, so the states can be processed in that order.  Back-pointers in
updates~\eqref{eq:dp-update} recover an optimum vertex set.

We verify that the state retains all information relevant to future
extensions.  For the chain containing a proposed new vertex, every previously
selected chain vertex points to it because its index is smaller.  For any other
chain $C_j$, test~\eqref{eq:transition-test} and the prefix property imply that
every earlier selected vertex of $C_j$ also points to the new vertex.  Hence
the appended sequence remains transitive.  This argument depends only on the
last selected vertex of each chain.  Consequently, any two partial sequences
realizing the same state have exactly the same valid one-vertex extensions,
which justifies retaining only the heavier one.

Conversely, let $A$ be any transitive subtournament.  Its unique topological
ordering is compatible with every chain: selected vertices from $C_i$ occur in
increasing chain order.  Reading this ordering from left to right produces
exactly transitions satisfying~\eqref{eq:transition-test}.  Thus the DP
represents every transitive subtournament and only transitive subtournaments.

There are $\prod_i(n_i+1)\leq(n+1)^b$ states.  Each state considers at most
$n$ candidates, each checked against at most $b$ last vertices, giving the
claimed running time.  The complement of a maximum-weight transitive
subtournament is a minimum-weight FVS.
\end{proof}

The chain partition is obtained by the standard bipartite-matching
implementation of Dilworth's theorem.  The same matching computation recovers
a maximum antichain.  Hence the exact subroutine is constructive and does not
enumerate transitive subtournaments.

The DP also strengthens what is known for one of the standard residual
classes used by approximation algorithms.

\begin{corollary}\label{cor:b7-free}
A minimum-weight FVS in a $\mathcal B_7$-free tournament can be computed in
time $O(n^7)$.
\end{corollary}

\begin{proof}
We claim that $\omega(\mcH(T))\leq6$.  Otherwise let $K$ be a seven-vertex
clique of $\mcH(T)$.  Since $T$ is $\mathcal B_7$-free,
$\tauv(T[K])\leq2$.  The complement of a minimum FVS in $T[K]$ is therefore a
transitive subtournament $A$ with $|A|\geq5$.  Every arc of $T[A]$ belongs to
the oriented triangle graph because $K$ is a clique.  Applying
Theorem~\ref{thm:blackbox}(ii) with $k=3$ produces a member of
$\mathcal B_7$ in $T$, a contradiction.  Lemma~\ref{lem:dp} with $b=6$ now
gives the result.
\end{proof}

\section{A constant-size strongly good obstruction}\label{sec:obstruction}

\begin{lemma}\label{lem:obstruction}
If $\omega(\mcH(T))\geq\Lk$, then, for fixed $k$, one can find in polynomial
time a set $Q\subseteq V(T)$ such that
\[
  |Q|\leq \Rk\,\tauv(T[Q]).
\]
\end{lemma}

\begin{proof}
Compute a maximum antichain of the poset $\overline{\mcH}(T)$ and choose
$\Lk$ of its vertices.  They form an $\Lk$-clique $K$ of $\mcH(T)$.  Since
$\Lk$ depends only on $k$, compute $\tauv(T[K])$ and a maximum transitive
subtournament of $T[K]$ by exhaustive enumeration.

If
\[
  \tauv(T[K])\geq \frac{k}{2k+1}\Lk,
\]
take $Q=K$.  Then
\[
  |Q|=\Lk\leq \frac{2k+1}{k}\tauv(T[K])
  =\Rk\,\tauv(T[Q]).
\]

Otherwise, let $A$ be a maximum transitive subtournament of $T[K]$.  Since a
minimum FVS is the complement of a maximum transitive subtournament,
\[
 |A|=\Lk-\tauv(T[K])
   >\frac{k+1}{2k+1}\Lk
   \geq \Pk.
\]
Both $|A|$ and $\Pk$ are integers, so $|A|\geq\Pk+1$.  Every pair of vertices
in $K$ is an edge of $\mcH(T)$, and hence every arc of the transitive
tournament $T[A]$ belongs to the oriented triangle graph.
Theorem~\ref{thm:blackbox}(ii) supplies a set $Q$ of $2k+1$ vertices with
$\tauv(T[Q])\geq k$.  Therefore
\[
  |Q|=2k+1\leq \frac{2k+1}{k}\tauv(T[Q])
  =\Rk\,\tauv(T[Q]).
\]

For fixed $k$, a suitable set in the last branch can be found without
constructivizing the proof of the black-box theorem: enumerate all
$(2k+1)$-subsets and test their FVS number by exhaustive enumeration.
\end{proof}

For the set $Q$ from Lemma~\ref{lem:obstruction}, define the cost vector
\[
 c_Q(v)=
 \begin{cases}
  1,&v\in Q,\\
  0,&v\notin Q.
 \end{cases}
\]

\begin{lemma}\label{lem:good}
The cost vector $c_Q$ is strongly $\Rk$-good.
\end{lemma}

\begin{proof}
Every FVS of $T$ restricts to an FVS of $T[Q]$, so
$\OPT(T,c_Q)\geq\tauv(T[Q])$.  Conversely, deleting every vertex outside $Q$
costs zero; after doing so, delete a minimum FVS inside $Q$.  Hence
$\OPT(T,c_Q)=\tauv(T[Q])$.  For every feasible FVS $S$,
\[
 c_Q(S)\leq |Q|\leq \Rk\,\tauv(T[Q])
 =\Rk\,\OPT(T,c_Q).
\]
\end{proof}

\section{The deterministic algorithm and its analysis}\label{sec:algorithm}

The algorithm is recursive on a tournament $T$ and a nonnegative rational
weight vector $w$.  The parameter $k$ is fixed throughout.  At each call, the
triangle graph can be constructed in $O(n^3)$ time by enumerating triples.

\begin{description}[leftmargin=3.4cm,style=nextline]
\item[Transitive case.]
If $T$ is transitive, return the empty set.

\item[Zero-weight case.]
If some vertex $v$ has weight zero, recursively solve $(T-v,w|_{V(T)\setminus
\{v\}})$.  If the returned set is already an FVS of $T$, return it; otherwise
add $v$ and return the enlarged set.

\item[Bounded-width case.]
Compute a minimum chain cover of $\overline{\mcH}(T)$.  If it has fewer than
$\Lk$ chains, invoke Lemma~\ref{lem:dp} and return the exact solution.

\item[Reduction case.]
Otherwise recover a maximum antichain from the matching computation and invoke
Lemma~\ref{lem:obstruction} to obtain $Q$.  Set
$\Delta=\min_{v\in Q}w(v)$ and $w'=w-\Delta c_Q$.  Recurse on $(T,w')$ and
return the recursive output.
\end{description}

The zero-weight case preserves feasibility: the recursive output is an FVS of
$T-v$, so any cycle remaining in $T$ after deleting it must contain $v$.
Adding $v$ eliminates all such cycles.  It also preserves the approximation
ratio, because an optimum solution for $T$ can have $v$ removed without
increasing its weight, and the added vertex has weight zero.

In the reduction case all weights are positive before the reduction, so
$\Delta>0$, the vector $w'$ is nonnegative, and at least one vertex of $Q$ has
weight zero under $w'$.  Thus the next recursive call either terminates or
enters the zero-weight case.  This gives strict termination under lexicographic
induction on
\[
  \bigl(|V(T)|,\;|\{v:w(v)>0\}|\bigr).
\]

\begin{theorem}\label{thm:main-k}
For every fixed integer $k\geq2$, the algorithm returns an FVS $S$ satisfying
\[
  w(S)\leq \left(2+\frac1k\right)\OPT(T,w).
\]
\end{theorem}

\begin{proof}
We use the lexicographic induction described above.  The transitive case is
exact, and the bounded-width case is exact by Lemma~\ref{lem:dp}.

In the zero-weight case, let $S'$ be the recursive solution.  If $O$ is an
optimum FVS of $T$, then $O\setminus\{v\}$ is an FVS of $T-v$ and has the same
weight because $w(v)=0$.  Consequently,
\[
  \OPT(T-v,w|_{V(T)\setminus\{v\}})\leq\OPT(T,w).
\]
The induction hypothesis bounds $w(S')$, and the optional addition of $v$
costs zero.

It remains to consider a reduction step.  Let $S$ be the recursive output, and
let $\Delta$ and $c_Q$ be as above.  The induction hypothesis gives
\[
  w'(S)\leq\Rk\,\OPT(T,w').
\]
Since $S$ is feasible and $c_Q$ is strongly $\Rk$-good,
Lemma~\ref{lem:good} gives
\[
  c_Q(S)\leq\Rk\,\OPT(T,c_Q).
\]
Let $O$ be an optimum FVS for the original weights $w$.  Since
$w=w'+\Delta c_Q$,
\[
\begin{aligned}
 \OPT(T,w')+\Delta\OPT(T,c_Q)
 &\leq w'(O)+\Delta c_Q(O)\\
 &=w(O)=\OPT(T,w).
\end{aligned}
\]
Therefore
\[
\begin{aligned}
 w(S)
 &=w'(S)+\Delta c_Q(S)\\
 &\leq \Rk\bigl(\OPT(T,w')+\Delta\OPT(T,c_Q)\bigr)\\
 &\leq \Rk\,\OPT(T,w).
\end{aligned}
\]
This completes the induction.
\end{proof}

\section{Running time and the approximation guarantee}\label{sec:complexity}

For fixed $k$, the values $\Pk$ and $\Lk$ are constants.  The triangle graph
is constructed in $O(n^3)$ time.  A minimum chain cover and a maximum
antichain in the poset $\overline{\mcH}(T)$ are recovered by bipartite matching.
Computing the FVS number of the constant-size clique $K$ costs
$2^{O(\Lk)}$.  Finding the set $Q$ in the second branch of
Lemma~\ref{lem:obstruction} costs
\[
  n^{2k+1}2^{O(k)}
\]
by enumerating all $(2k+1)$-subsets.  The terminal bounded-width call costs
$O_k(n^{\Lk})$ because the number of chains is at most $\Lk-1$.

A reduction creates a zero-weight vertex, which is removed before another
reduction can occur.  Hence there are at most $n$ reductions and at most
$2n+1$ recursive calls.  A uniform upper bound on the arithmetic operation
count is therefore
\[
  2^{O(\Lk)}n^{\Lk+O(1)}=n^{2^{O(k)}}.
\]

All arithmetic is exact over the rationals.  To justify the bit complexity,
put the input weights over a common denominator, whose bit length is at most
the total input encoding length.  Every residual weight is an integer linear
combination of the original numerators.  Each reduction is an elementary
subtraction, so after at most $n$ reductions the coefficient magnitudes are at
most exponential in $n$ and their bit lengths are $O(n)$.  Thus every residual
weight has polynomial encoding length.  The total running time is
\[
  n^{2^{O(k)}}\operatorname{poly}(\langle w\rangle).
\]

\begin{corollary}\label{cor:epsilon}
For every fixed $\varepsilon>0$, weighted $\TFVS$ has a deterministic
approximation algorithm with ratio at most $2+\varepsilon$.  Its running time is
\[
  n^{2^{O(1/\varepsilon)}}\operatorname{poly}(\langle w\rangle).
\]
\end{corollary}

\begin{proof}
Choose $k=\max\{2,\lceil1/\varepsilon\rceil\}$.  Then
$1/k\leq\varepsilon$, and Theorem~\ref{thm:main-k} gives the claimed ratio and
running time.
\end{proof}

\section{Discussion}

The dependence on $k$ is intentionally not optimized.  The exponential
threshold $\Pk=2^{k-1}$ comes from the structural theorem for triangle graphs,
while the bounded-width DP contributes the exponent $\Lk=O(2^k)$.  Improving
the dependence on $\varepsilon$ would require either a stronger obstruction
theorem or a more economical exact algorithm in the residual class.

The proof is specific to tournaments.  The triangle-graph structural theorem
also applies to quasi-transitive digraphs, but extending the dynamic program
requires a direct treatment of non-edges and of the tournament embedding used
in that theorem.  We therefore make no claim for the larger class here.

\bibliographystyle{alpha}
\bibliography{references}

@article{aprile2023,
  author  = {Aprile, Manuel and Drescher, Matthew and Fiorini, Samuel and Huynh, Tony},
  title   = {A 7/3-Approximation Algorithm for Feedback Vertex Set in Tournaments via {Sherali--Adams}},
  journal = {Discrete Applied Mathematics},
  volume  = {337},
  pages   = {149--160},
  year    = {2023},
  doi     = {10.1016/j.dam.2023.04.016}
}

@article{baryehuda1981,
  author  = {Bar-Yehuda, Reuven and Even, Shimon},
  title   = {A Linear-Time Approximation Algorithm for the Weighted Vertex Cover Problem},
  journal = {Journal of Algorithms},
  volume  = {2},
  number  = {2},
  pages   = {198--203},
  year    = {1981},
  doi     = {10.1016/0196-6774(81)90020-2}
}

@article{cai2001,
  author  = {Cai, Mao-cheng and Deng, Xiaotie and Zang, Wenan},
  title   = {An Approximation Algorithm for Feedback Vertex Sets in Tournaments},
  journal = {SIAM Journal on Computing},
  volume  = {30},
  number  = {6},
  pages   = {1993--2007},
  year    = {2001},
  doi     = {10.1137/S0097539798338163}
}

@article{dilworth1950,
  author  = {Dilworth, Robert P.},
  title   = {A Decomposition Theorem for Partially Ordered Sets},
  journal = {Annals of Mathematics},
  volume  = {51},
  number  = {1},
  pages   = {161--166},
  year    = {1950},
  doi     = {10.2307/1969503}
}

@inproceedings{ghorbani2026,
  author    = {Ghorbani, Ebrahim and Mnich, Matthias},
  title     = {A 9/4-Approximation for Directed Feedback Vertex Sets in Quasi-Transitive Digraphs},
  booktitle = {53rd International Colloquium on Automata, Languages, and Programming (ICALP 2026)},
  series    = {Leibniz International Proceedings in Informatics (LIPIcs)},
  volume    = {374},
  pages     = {96:1--96:16},
  publisher = {Schloss Dagstuhl -- Leibniz-Zentrum f{\"u}r Informatik},
  year      = {2026},
  doi       = {10.4230/LIPIcs.ICALP.2026.96},
  url       = {https://doi.org/10.4230/LIPIcs.ICALP.2026.96}
}

@article{lokshtanov2021,
  author  = {Lokshtanov, Daniel and Misra, Pranabendu and Mukherjee, Joydeep and Panolan, Fahad and Philip, Geevarghese and Saurabh, Saket},
  title   = {2-Approximating Feedback Vertex Set in Tournaments},
  journal = {ACM Transactions on Algorithms},
  volume  = {17},
  number  = {2},
  pages   = {11:1--11:14},
  year    = {2021},
  doi     = {10.1145/3446969}
}

@inproceedings{mnich2016,
  author    = {Mnich, Matthias and Vassilevska Williams, Virginia and V{\'e}gh, L{\'a}szl{\'o} A.},
  title     = {A 7/3-Approximation for Feedback Vertex Sets in Tournaments},
  booktitle = {24th Annual European Symposium on Algorithms (ESA 2016)},
  series    = {Leibniz International Proceedings in Informatics (LIPIcs)},
  volume    = {57},
  pages     = {67:1--67:14},
  publisher = {Schloss Dagstuhl -- Leibniz-Zentrum f{\"u}r Informatik},
  year      = {2016},
  doi       = {10.4230/LIPIcs.ESA.2016.67}
}

\end{document}